\documentclass[letterpaper]{article} 
\usepackage{aaai2026}  
\usepackage{times}  
\usepackage{helvet}  
\usepackage{courier}  
\usepackage[hyphens]{url}  
\usepackage{graphicx} 
\usepackage{natbib}  
\usepackage{caption} 
\usepackage{subfigure}
\usepackage{multirow}
\usepackage{algorithm}
\usepackage{amsmath}
\usepackage{booktabs}
\usepackage{dsfont}

\usepackage{amsmath,amsfonts,bm}

\def\eqref#1{equation~\ref{#1}}

\def\1{\bm{1}}

\DeclareMathAlphabet{\mathsfit}{\encodingdefault}{\sfdefault}{m}{sl}
\SetMathAlphabet{\mathsfit}{bold}{\encodingdefault}{\sfdefault}{bx}{n}

\usepackage{amsthm,amsmath,amsfonts,bm,amssymb}
\usepackage{physics,graphicx}
\usepackage{xcolor}
\usepackage[shortlabels]{enumitem}
\newtheorem{theorem}{Theorem}

\newtheorem{definition}{Definition}

\theoremstyle{plain}

\theoremstyle{definition}

\theoremstyle{remark}

\usepackage{enumitem}
\definecolor{tan}{rgb}{0.937, 0.902, 0.843}
\usepackage{newfloat}
\usepackage{listings}
\DeclareCaptionStyle{ruled}{labelfont=normalfont,labelsep=colon,strut=off} 
\floatstyle{ruled}
\newfloat{listing}{tb}{lst}{}
\floatname{listing}{Listing}
\title{Graph Domain Adaptation via Homophily-Agnostic Reconstructing Structure}
\author{
    Ruiyi Fang\textsuperscript{\rm 1},
    Shuo Wang\textsuperscript{\rm 2},
    Ruizhi Pu\textsuperscript{\rm 1}\thanks{Corresponding author},
    Qiuhao Zeng\textsuperscript{\rm 1},
    Hao Zheng\textsuperscript{\rm 3},
    Ziyan Wang\textsuperscript{\rm 1},
    Jiale Cai\textsuperscript{\rm 1},
    Zhimin Mei\textsuperscript{\rm 1},
    Song Tang\textsuperscript{\rm 4},
    Charles Ling\textsuperscript{\rm 1}, 
    Boyu Wang\textsuperscript{\rm 1}
}

\affiliations{
    \textsuperscript{\rm 1} Department of Computer Sceince, Western University\\
    \textsuperscript{\rm 2} School of Computer Science and Engineering, University of Electronic Science and Technology of China\\
    \textsuperscript{\rm 3} School of Computer Science and Engineering, Central South University\\
    \textsuperscript{\rm 4} Institute of Machine Intelligence, University of Shanghai for Science and Technology\\

}

\usepackage{bibentry}
\usepackage{amssymb,amsmath,amsthm}
\usepackage{mathtools}
\usepackage{amsfonts}
\usepackage{multirow}
\usepackage{diagbox}

\begin{document}

\maketitle

\begin{abstract}
Graph Domain Adaptation (GDA) transfers knowledge from labeled source graphs to unlabeled target graphs, addressing the challenge of label scarcity. However, existing GDA methods typically assume that both source and target graphs exhibit homophily, leading existing methods to perform poorly when heterophily is present. Furthermore, the lack of labels in the target graph makes it impossible to assess its homophily level beforehand. To address this challenge, we propose a novel homophily-agnostic approach that effectively transfers knowledge between graphs with varying degrees of homophily. Specifically, we adopt a divide-and-conquer strategy that first separately reconstructs highly homophilic and heterophilic variants of both the source and target graphs, and then performs knowledge alignment separately between corresponding graph variants. Extensive experiments conducted on five benchmark datasets demonstrate the superior performance of our approach, particularly highlighting its substantial advantages on heterophilic graphs.
\end{abstract}

\section{Introduction}
Graphs are pervasive and have been widely used in numerous real-world scenarios, such as social networks, traffic networks, and recommendation systems \cite{liu2021self}. However, they are often constrained by label scarcity, as annotating structured data remains challenging and costly \cite{AGC,liu2022multilayer}. To address this challenge, graph domain adaptation (GDA) has emerged as an effective paradigm for transferring knowledge from labeled source graphs to an unlabeled target graph. Existing GDA methods built on traditional GNNs typically rely on message passing mechanisms that assume homophily \cite{wu2020unsupervised,xie2023contrastive}. Facing heterophilic graphs, existing approaches mainly suffer from two limitations. Firstly, the local neighbors in a graph are nodes that are proximally located, while semantically similar nodes might be far apart on a heterophilic graph \cite{H2GCN,xie2025robust}. Thus, existing methods struggle to capture long-range divergence from distant nodes and to distinguish between homophilic and heterophilic neighbors, which convey different label information, leading to the fact that they perform well on homophilic graphs often exhibit degraded performance on heterophilic graphs \cite{fang2025benefits}.\\

\begin{figure}[t]
    \centering
    \includegraphics[width=.9\linewidth]{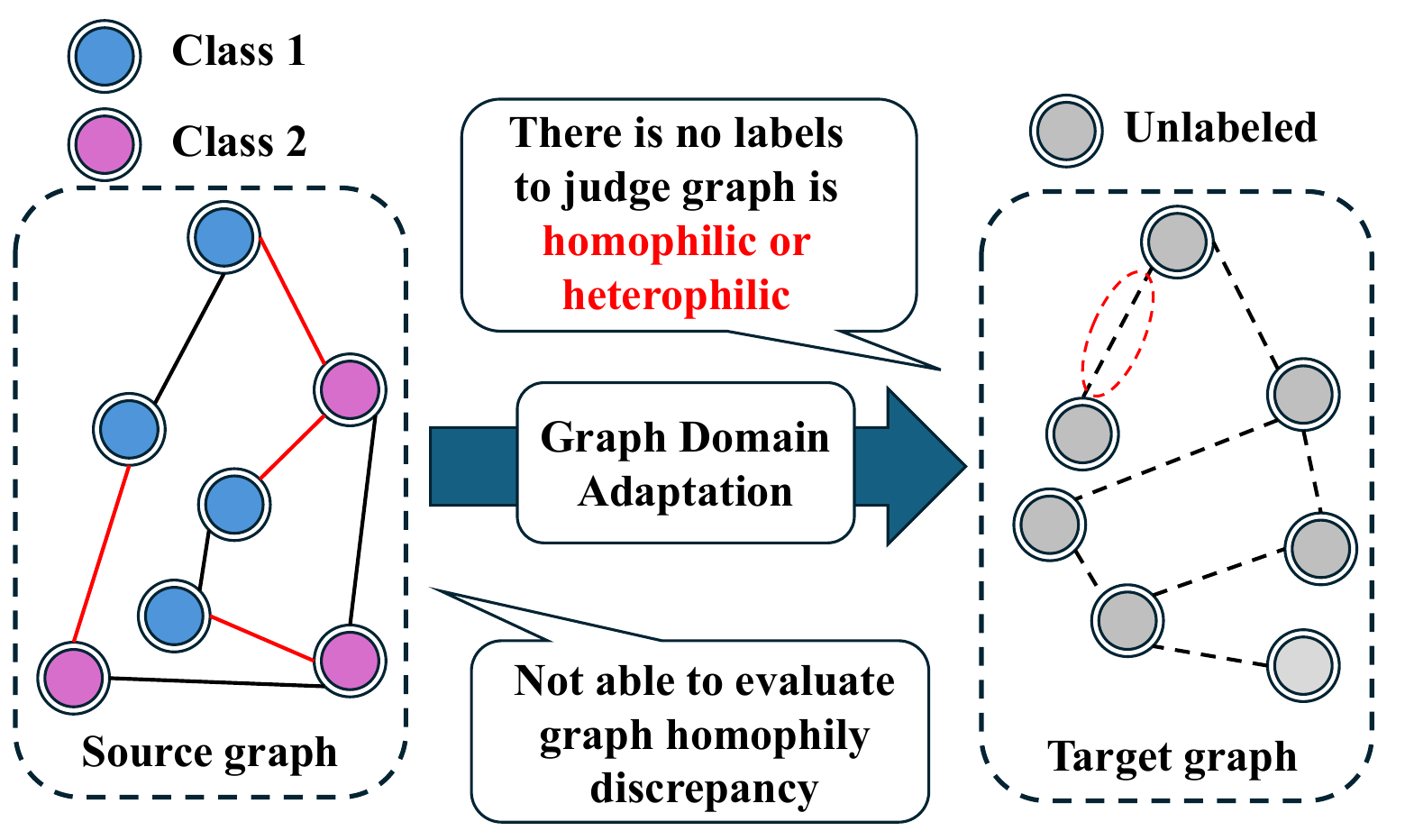}    
    \caption{Illustration of the graph domain adaptation task (best viewed in color). Given a labeled source graph (color indicates node label) and an unlabeled target graph, there are no labels for us to judge whether the target graph is homophilic or heterophilic.}
    \label{fig1}
\end{figure}

\begin{figure*}[t]
    \centering
    \subfigure[Framework Overview]{
        \includegraphics[width=0.59\textwidth]{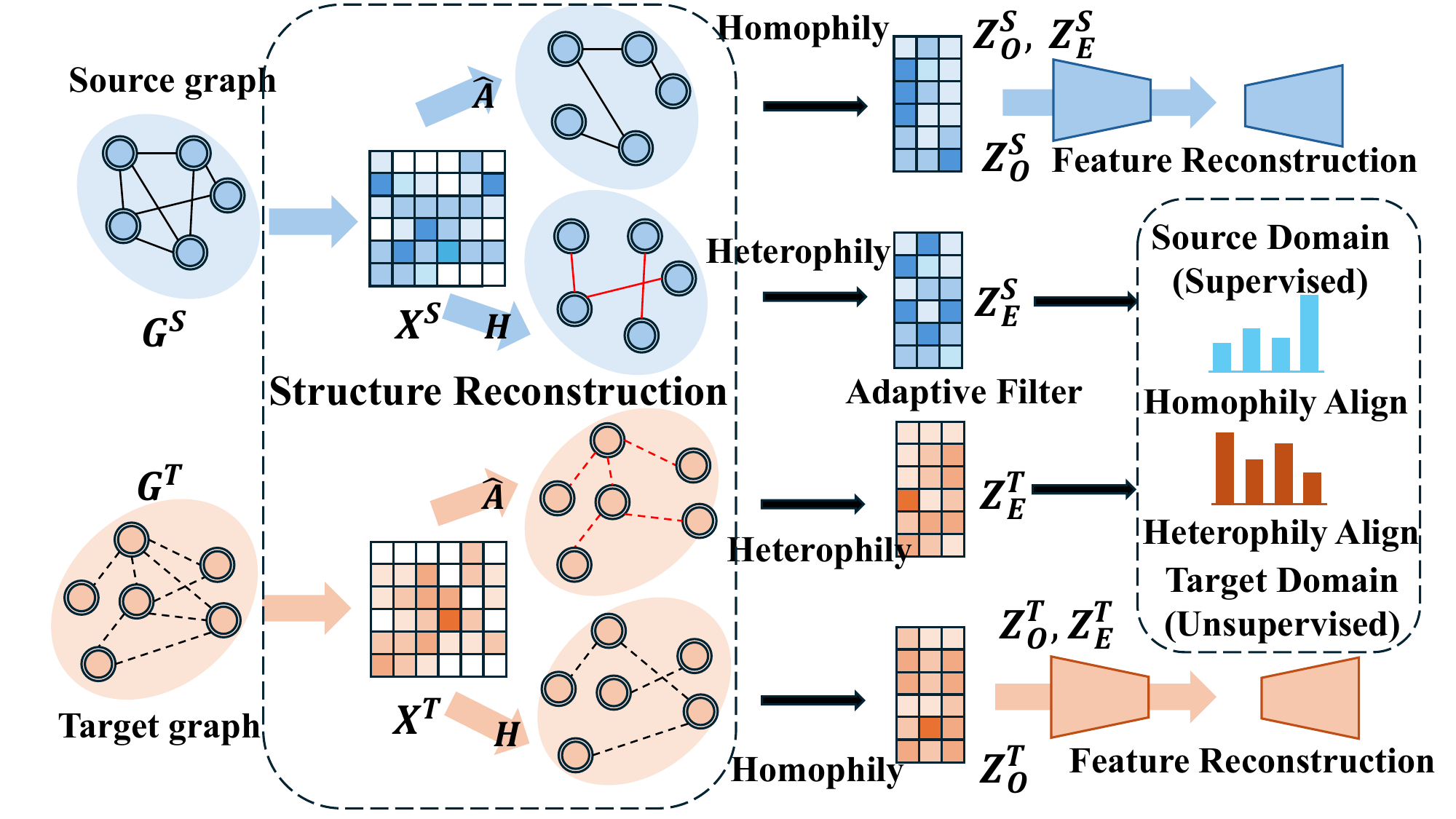}
    }
        \hspace{0.01\textwidth}
    \subfigure[Homophily Distribution of Original and reconstructed graph]{
        \includegraphics[width=0.36\textwidth]{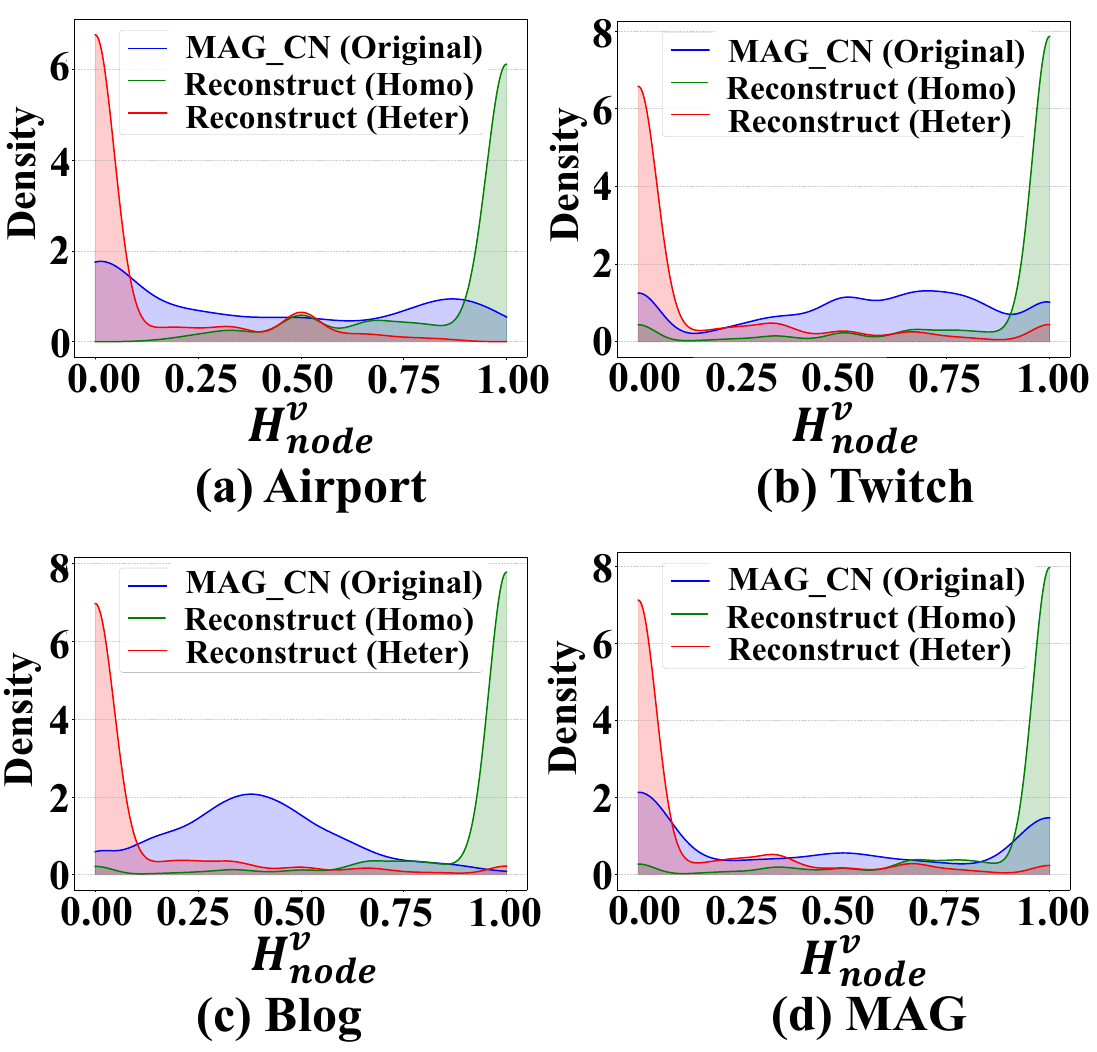}
    }
    \caption{(a) An overview of our method. RSGDA reconstructs structure to obtain homophilic and heterophilic variant graphs, where the final goal is to minimize the graph distribution shift through separately aligning homophily and heterophily. (b) The blue line represents the graph homophily distribution in five benchmarks. The red and green lines demonstrate that after structural reconstruction, RSGDA can effectively separate reconstructed homophilic and heterophilic variants. Homophily means that similar nodes are prone to connect to each other.}
    \label{framework}
\end{figure*}

Some graph clustering methods have been proposed to deal with heterophily by expanding neighbor fields\cite{BernNet, ACMNN, JacobiNet} or refining the GNN architectures\cite{FAGCN, DMP}.
Though the aforementioned methods can not be direct used in GDA due to there exists two critical problems: 1) The training of customized network, the learning of adaptive filter, and graph rewiring \cite{DHGR} rely on labeled samples, which makes them not be applicable to unsupervised GDA task (as shown in Fig. \ref{fig1}). 2) The above methods process graph signals could damaging entanglement of homophilic and heterophilic information hinders their ability to mitigate distributional divergence \cite{fanghomophily}. \\

 \textbf{For unsupervised GDA, the first and foremost challenge we face is that there is no labels for us to judge whether a graph is homophilic or heterophilic.} Therefore, it is not practical to develop individual models to handle homophilic and heterophilic graphs separately for GDA. Moreover, it is also subjective to simply classify a graph as homophilic or heterophilic since real-world graph data could have various levels of homophily. To address this pivotal challenge, we propose a holistic framework named reconstruction structure GDA (RSGDA) to accommodate real-world graphs in Fig. \ref{framework} (a). We first construct new homophilic and heterophilic graphs to explore both homophily and heterophily information. Notably, this graph reconstruction process is fully unsupervised and broadly applicable. Fig.~\ref{framework} (b) illustrates the effectiveness of the reconstruction process in RSGDA, which enables the extraction of varying levels of homophily by generating separate homophilic and heterophilic graph variants. An adaptive filter is employed to capture both homophilic and heterophilic graph signals, and the resulting features are fed into a homophily-agnostic alignment network that aligns these signals separately. We summarize our contributions as follows:
\begin{itemize}
\item We firstly deal with various levels of heterophilic graphs through constructing highly homophilic and heterophilic variants in GDA with homophily-agnostic graph.
\item We design an adaptive filtering and alignment mechanism that captures and aligns both homophilic and heterophilic signals of the data without labels.
\item Extensive experiments on homophilic and heterophilic benchmarks demonstrate the promising performance of our proposed method.
\end{itemize}

\section{Related work}
\subsection{Heterophilic Graph Learning}
Heterophilic structure is prevalent in practice, from personal relationships in daily life to chemical and molecular scientific study. Developing powerful heterophilic GNN models is a hot research topic. \cite{lim1, lim2} provide general benchmarks for heterophilic graph learning. In addition, many methods have been proposed to revise GNNs for heterophilic graphs. \cite{DMP} specifies propagation weight for each attribute to make GNNs fit heterophilic graphs and \cite{GloGNN} explores the underlying homophilic information by capturing the global correlation of nodes. \cite{H2GCN} enlarges receptive field via exploring high-order structure. \cite{GPR} adaptively combines the representation of each layer, and \cite{GCNII} integrates embeddings from different depths with a residual operation. Recent studies ~\cite{pu2025fedelr,luan2022revisiting} reveal that graph homophilc and heterophilic patterns impact graph clustering performance between the test and training sets. However, the above methods failed to explore the role of homophily in GDA and minimize their discrepancy in many aspects~\cite {kang2024cdc,li2024pc,zhuo2024improving,pu2025leveraging,zhuo2023propagation,zhuo2024unified,zhuo2025dualformer}.\\

 \subsection{Graph Domain Adaptation}
Recent Domain Adaptation works have differences from GDA methods ~\cite{li2024phrase,chen2024text,li2025let,li2024object,li2025simplified, wang2025multi,huang2025enhancing, huangnodes,huang2023robust}. Identifying the differences between the target and source graphs in GDA is crucial. For graph-structured data, several studies have explored cross-graph knowledge transfer using graph domain adaptation (GDA) methods~\cite{shen2019network,dai2022graph,shi2024graph, shen2024beyond, wang2025cooperation}. Some graph information alignment-based methods~\cite{shen2020adversarial,shen2020network,yan2020graphae,shen2023domain,zhuocloser,zheng2025test,liu2025test,zheng2024structure,xu2025unraveling} adapt graph source node label information by integrating global and local structures from both nodes and their neighbors. 
UDAGCN~\cite{wu2020unsupervised} introduces a dual graph convolutional network that captures both local and global knowledge, adapting it through adversarial training. Furthermore, ASN and GraphAE~\cite{zhang2021adversarial,guo2022learning} consider extracting and aligning graph specific information like node degree and edge shift, enabling the extraction of shared features across networks. SOGA~\cite{mao2021source} is the first to incorporate discriminability by promoting structural consistency between target nodes of the same class, specifically for source-free domain adaptation (SFDA) on graphs.
SpecReg~\cite{you2022graph} applies an optimal transport-based GDA bound and demonstrates that revising the Lipschitz constant of GNNs can enhance performance through spectral smoothness and maximum frequency response. 
JHGDA~\cite{shi2023improving} tackles hierarchical graph structure shifts by aggregating domain discrepancies across all hierarchy levels to provide a comprehensive discrepancy measure.
ALEX~\cite{yuan2023alex} creates a label-shift-enhanced augmented graph view using a low-rank adjacency matrix obtained through singular value decomposition, guided by a contrasting loss function. SGDA~\cite{qiao2023semi} incorporates trainable perturbations (adaptive shift parameters) into embeddings via adversarial learning, enhancing source graphs and minimizing marginal shifts. PA~\cite{liu2024pairwise} mitigates structural and label shifts by recalibrating edge weights to adjust the influence among neighboring nodes, addressing conditional structure shifts effectively. GAA~\cite{qiao2023semi} separately extracts and aligns graph attributes and structural information through a feature graph.

\section{Methodology}
\subsection{Notation}
An graph $ G = \left\{\mathcal{V}, \mathcal{E}, A, X, Y\right\} $ consists of a set of nodes $ \mathcal{V} $ and edges $ \mathcal{E} $, along with an adjacency matrix $ A $, a feature matrix $ X $, and a label matrix $ Y $. The adjacency matrix $ A \in \mathbb{R}^{N \times N} $ encodes the connections between $ N $ nodes, where $ A_{ij} = 1 $ indicates an edge between nodes $ i $ and $ j $, and $ A_{ij} = 0 $ means the nodes are not connected. The feature matrix $ X \in \mathbb{R}^{N \times d} $ represents the node features, with each node described by a $ d $-dimensional feature vector. Finally, $ Y \in \mathbb{R}^{N \times C} $ contains the labels for the $ N $ nodes, where each node is classified into one of $ C $ classes. Thus, the symmetric normalized adjacency matrix is $\tilde{A}=(D+I)^{-\frac{1}{2}} (A+I) (D+I)^{-\frac{1}{2}}$ and the normalized Laplacian matrix is $\tilde{L}=I-\tilde{A}$. 
In this work, we explore the task of node classification in an unsupervised setting, where both the node feature matrix $X$ and the graph structure $A$ are given before learning. Now we can define source graph $G^S=\left\{\mathcal{V}^S, \mathcal{E}^S, A^S, X^S, Y^S\right\}$ and target graph $G^T=\left\{\mathcal{V}^T, \mathcal{E}^T, A^T, X^T\right\}$.  

\textbf{Local node homophily ratio} is a widely used metric for quantifying homophilic and heterophilic patterns. It is defined as the proportion of a node's neighbors that share the same label as the node~\cite{MiaoZWL0024}.
It is formally defined as 
\begin{equation}
H_{\mathrm{node}}^v=\frac{\left|\left\{u \mid u \in N_v, y_u=y_v\right\}\right|}{\left|N_v\right|}
\end{equation}
where where $N_v$ denotes the set of one-hop neighbors of node $v$ and $y_i$ is the node $i$ label. 

\begin{figure}[t]
    \centering
    \subfigure[WebKB]{
        \includegraphics[width=.47\linewidth]{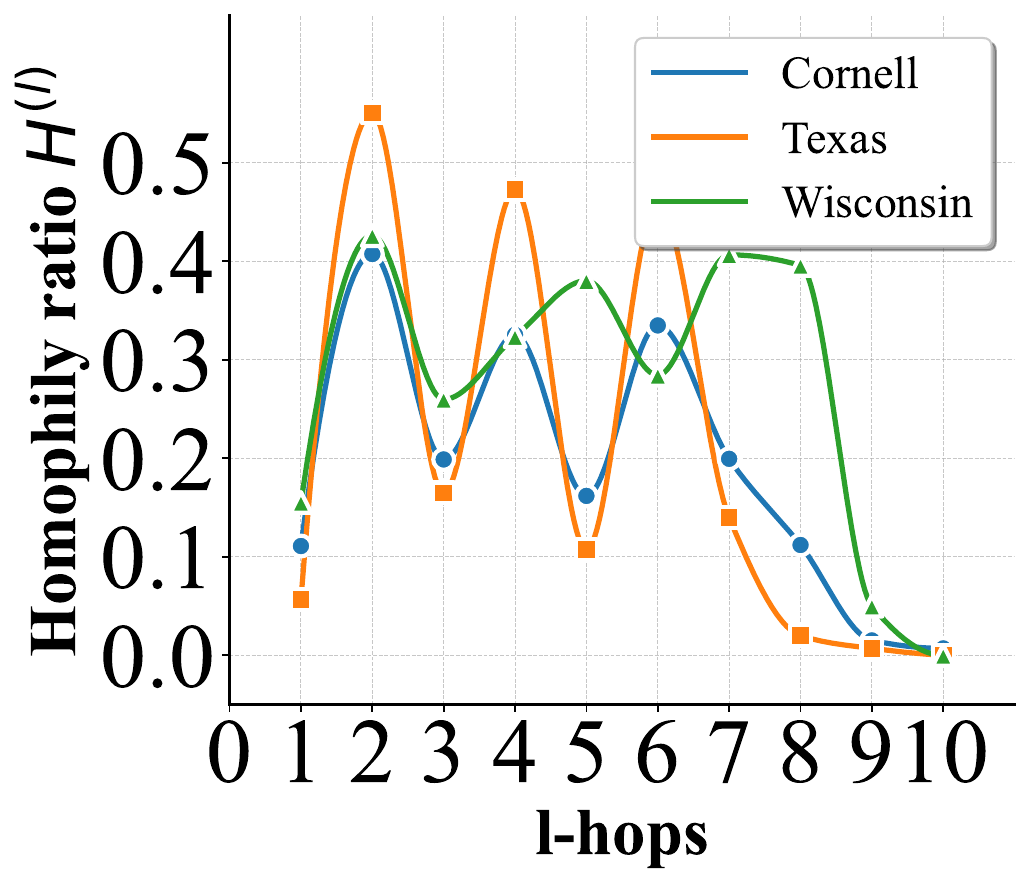}
    }
    \subfigure[Airport]{
        \includegraphics[width=.47\linewidth]{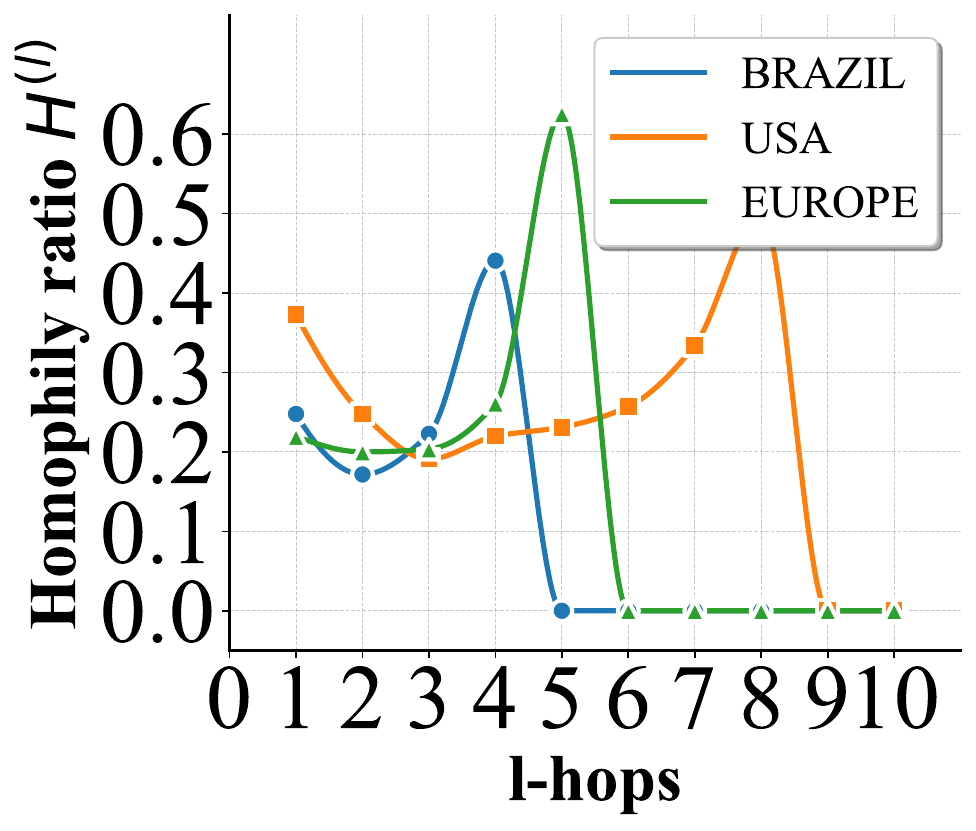}
    }
    \caption{This show the changes of homophily ratio in homophilic and heterophilic graphs with different hops. Particularly, the nodes in different hops sharing different neighbors have various homophily ratios.}
    \label{2hop1}
\end{figure}

\subsection{Structure Reconstruction} \label{GC}
	
In practice, we can't know whether a given graph is homophilic or not in unsupervised GDA tasks. Hence, separately developing homophilic and heterophilic methods is unrealistic. Moreover, real graphs always contain both homophilic and heterophilic nodes. To have a holistic model, we develop a structure reconstruction approach. Specifically, we construct a heterophilic graph and a homophilic graph from the original graph.

\subsubsection{Homophilic Graph Construction}
In practice, even the homophilic graph doesn't have a homophily score of $1$, i.e., there exist some heterophilic nodes in the homophilic graph. Thus, we could further improve the homophily level of graph by minimizing the distances among adjacent nodes, which is formulated as:
$$
\min _{\hat{A}_{i:}} \sum_{j=1}^N \hat{A}_{i j}\left\|x_i-x_j\right\|^2,
$$
where $\hat{A}_{i:}$ represents the $i$-th row of $\hat{X}$. To avoid the trivial solution $\hat{X}=I$, we rewrite the above equation as: 
$$
\min _{ \hat{A}_{i:}} \sum_{j=1}^N \hat{A}_{i j}\left\|x_i-x_j\right\|^2+\hat{A}_{i j}^2.
$$
It’s clear that the graph $\hat{A}$ will be more homophilic when edges are defined by nodes sharing high similarity~\cite{pan2023beyond,xie2025one}. As shown in Fig.~\ref{2hop1}, the homophily ratio $H^{(l)}$ varies significantly across different hop levels, suggesting the need to extract structural information at multiple scales. The homophily ratio is defined as $H^{(l)}(G) = \frac{1}{|A_{ij}^{(l)}|} \sum_{A_{ij}^{(l)} > 0} \mathds{1}(y_i = y_j)$, where $l$ denotes the number of hops and the exponent of the adjacency matrix $A$ \cite{HM}. Consequently, the message propagation path in the original ${A}$ could be incorrect when homophily ratio changes of a node are dissimilar due to different hop neighbors. Furthermore, using 1-hop structural information may limit the effectiveness of message propagation. Based on above observation, we empirically design a term to integrate the various hops neighbor relation, i.e., we enforce that all $l$-hop neighboring nodes are in the set of 1-hop neighborhood. Let $\|x_i - x_j\|^2=F_{ij}$ as feature distance, then we can construct a homophilic graph $\hat{A}$ by solving the following optimization problem:
    \begin{equation} \label{pr2}
    \begin{gathered}
    \min _{\hat{A}_{ij}} \hat{A}_{i j} F_{ij}+\hat{A}_{i j}^2+\left(\hat{A}_{i j}^{(l)}-\hat{A}_{i j}\right)^2, \\
    \text { s.t. } \hat{A}_{i j}>0, \sum_{j=1}^N \hat{A}_{i j}=1,
    \end{gathered}
    \end{equation}
where $\hat{A}^{(l)}$ is the l-hop graph, i.e., $\hat{A}^{(2)}=\hat{A}\times \hat{A}$.
\paragraph{Construction Optimization}
In practice, we use $l = 2$ as a representative example, though $l$ can be set to other values by repeating the same procedure. We begin by initializing $\hat{A}$ with $A$, and then reformulate problem (\ref{pr2}) using its corresponding Lagrangian function:
	\begin{equation} \label{hgl_lagrangian}
		\begin{aligned}
			\underset{\hat{A}_{i:}}{\operatorname{min}}\sum_{j=1}^N[&\hat{A}_{ij}F_{ij}+(\hat{A}_{ij}^{(2)}-\hat{A}_{ij})^2+\hat{A}_{ij}^2]\\
			&-\sum _{j=1}^N\lambda_{1j}\hat{A}_{ij}-\lambda_{2i}(\sum_{j=1}^N\hat{A}_{ij}-1).
		\end{aligned}
	\end{equation}
The derivative of Eq. (\ref{hgl_lagrangian}) w.r.t. $\hat{A}_{ij}$ is 
	\begin{equation} \label{hgl_derivative}
		\begin{aligned}
			&F_{ij}+2(\hat{A}_{ii}+\hat{A}_{jj}-1)(\hat{A}_{ij}^{(2)}-\hat{A}_{ij})+2\hat{A}_{ij}-\lambda_{1j}-\lambda_{2j}\\
			&+\sum_{f\neq j}^N2\hat{A}_{jf}[\hat{A}_{ij}\hat{A}_{jf}+(\hat{A}_{if}^{(2)}-\hat{A}_{ij}\hat{A}_{jf})-\hat{A}_{if}].
		\end{aligned}
	\end{equation}
Remove the self-loop on graph, i.e., let $\hat{A}_{ii}=0$. $\hat{A}$ and $\hat{A}^{(l)}$ can be regarded as constants at the last iteration. By introducing $C_f$, 
	$$C_f=\hat{A}_{if}^{(l)}-\hat{A}_{ij}\hat{A}_{jf}-\hat{A}_{if},i\neq f\\ 
    $$  
we rewrite Eq. (\ref{hgl_derivative}) as:
	\begin{equation} \label{hgl_derivative_simplify}
		F_{ij}-2(\hat{A}_{ij}^{(2)}-\hat{A}_{ij})+2\hat{A}_{ij}-\lambda_{1j}-\lambda_{2i}+\sum_{f\neq j}2\hat{A}_{jf}^2\hat{A}_{ij}+2\hat{A}_{jf}C_f.
	\end{equation} 
	
Based on previous work inference~\cite{pan2023beyond}, to obtain $\hat{A}_{ij}$, we need to compute $\lambda_{2j}$ by solving this optimization problem:
	\begin{equation} \label{hgl_lambda}
		\underset{\lambda_{2j}}{\operatorname{min}} \sum_{j=1}^N  [\frac{2\hat{A}_{ij}^{(2)}+\lambda_{2j}-F_{ij}-2\sum_{f \neq j}\hat{A}_{jf}C_f}{2(2+\sum_{f \neq j}\hat{A}_{jf}^2)}]_+ -1.
	\end{equation}
where $[\bullet]_+$ operator means $\operatorname{max}(\bullet,0)$. 
Equation (\ref{hgl_lambda}) can be solved using either a gradient descent algorithm or by formulating it as a linear programming (LP) problem. In this optimization process, $l$ can serve as a trade-off preprocessing parameter, and the resulting solution yields the source and target homophily graph adjacency matrices, denoted as ${A}_O^S$ and ${A}_O^T$, respectively.

\subsubsection{Heterophilic Graph Construction}
Firstly, we select the nodes that are far away from each other in both feature space and structure space as negative pairs, which prevents us from false negative pairs. Specifically, we use complementary graphs of the similarity graph and the topology graph to construct a heterophily graph. The procedure is formulated as follows:
	\begin{equation}
		\begin{aligned}
			\bar{S} &= 1.-S,\\ 
			\bar{A} &=1.-\tilde{A},\\ 
			H &=\bar{S} \odot \bar{A},
		\end{aligned}
	\end{equation}
The similarity matrix $S$ is computed using the cosine similarity of node features, capturing the closeness between nodes in the feature space. The symbol $\odot$ denotes the Hadamard product, which enables the identification of non-neighboring relationships in both the feature and topology spaces. In homophilic graphs, nodes of the same class are typically adjacent, while non-adjacent nodes (as reflected in the complementary graph) are more likely to belong to different classes. In contrast, heterophilic graphs often connect nodes from different classes, leading to dissimilar adjacent nodes \cite{pan2023beyond}. To construct a heterophilic graph, it is therefore reasonable to select adjacent nodes from both the complementary similarity matrix $\bar{S}$ and the complementary adjacency matrix $\bar{A}$. The resulting reconstructed graph $H$ may be dense; thus, we retain only the top five edges per node, corresponding to the five most dissimilar nodes. In this manner, we obtain the source heterophilic graph adjacency matrix $A_E^S$ and the target heterophilic graph adjacency matrix $A_E^T$.
	
\subsection{Graph Filtering}
 
Based on the assumption that the graph signal should be smooth, i.e., the neighbor nodes tend to be similar, a low-pass filter has been used to obtain smoothed representations\cite{MCGC}. 
One typical low-pass filter  \cite{AGC} can be formulated as:
	\begin{equation} \label{high_pass_filter}
		Z = (I-\frac{1}{2}L)^k X,
	\end{equation}
where $Z$ is the filtered representation, $k$ is the order of graph filtering.\\
However, Eq. (\ref{high_pass_filter}) could be ineffective, resulting from its heavy dependence on the raw topological graph, which could be noisy and incomplete. Additionally, low-pass filtering neglects the high-frequency components in data, which leads to information loss and inferior performance. This would be worse for a heterophilic graph, where high-frequency information plays a critical role\cite{MCGC}. Since it is impossible to know whether a given graph is homophilic or not in unsupervised learning, it's necessary to design a generic filter to handle different types of graphs. To this end, we design an adaptive filter for graph data as follows:
	\begin{equation} \label{H_filter}
		Z_E = \gamma (\frac{1}{2}L_E)^k X ,
	\end{equation}
    \begin{equation} \label{L_filter}
    Z_O = (1-\gamma) (I - \frac{1}{2}L_O)^{k} X,
    \end{equation}
where $\gamma$ is a learnable parameter balancing homophilic and heterophilic representations, $L_E = I - A_E$ and $L_O = I - A_O$ are the normalized Laplacian matrices of reconstructed homophilic graph $A_O$ and heterophilic graphs $A_E$. Note that our adaptive filter is not simply combining a low-pass and high-pass filter; we apply newly constructed graphs rather than the raw graph, which often has low quality. Finally, we obtain source representation $Z^S_E$ and $Z^S_O$ and target representation $Z^T_E$ and $Z^T_O$. \\

\subsection{Homophily-agnostic Alignment Network}
In this section, we introduce our proposed homophily-agnostic alignment Network to address the nodes GDA task. As shown in Fig. \ref{framework} (a), RSGDA contains two unshared encoders and a decoder, and all of them are MLPs. Different from previous methods, which learn and align representations in only one space, two encoders $E_{\theta_E}$ and $E_{\theta_O}$  are utilized to map filtered features and structural graph into homophilic signals and heterophilic signals, respectively. The encoder is applied to preserve different levels of original homophilic information. The obtained representations are:
	\begin{equation}
		\begin{aligned}
			Z_E = E_{\theta_E}(Z_E),\\
			Z_O = E_{\theta_O}(Z_O).
		\end{aligned}
	\end{equation} 
Moreover, to alleviate representation collapse, i.e., representations of all nodes tend to be the same, we add a correlation reduction item to prevent it \cite{BarlowTwins}.
\begin{equation}
    \begin{aligned}
        \mathcal{L}_{CR} =\frac{1}{d^2} \sum_{i=1}^d\left(\mathbf{K}_{i i}-1\right)^2+\frac{1}{d^2-d} \sum_{i=1}^d \sum_{j \neq i}\mathbf{K}_{i j}^2,
    \end{aligned}
\end{equation}
where $d$ is the dimension of node attribute, $M$ is the similarity of corresponding nodes in two encoders, $$K_{ij}=\frac{Z_E{}_{i}^{\top}Z_O{}_{j}}{\|Z_E{}_{i}\|\cdot\|Z_O{}_{j}\|}.$$ Afterwards, $Z_E$ and $Z_O$ are concatenated as $Z$.

Decoder is employed to obtain reconstructed features $\bar{Z}$. We adopt the Scaled Cosine Error as the objective of reconstruction \cite{GraphMAE}, which can down-weight easy samples’ contribution by controlling the sharpening parameter $\beta$ in training:
\begin{equation}
    \mathcal{L}_{\mathrm{RE}}=\sum^N_{i=1 }\left(1-\frac{Z_i^\top \bar{Z}_i}{\left\|Z_i\right\| \cdot\left\|\bar{Z}_i\right\|}\right)^\beta,  
\beta \geq 1.
\end{equation}

After the above steps, source and target are processed separately, we finally obtain $Z_{E}^S $, $ Z_{O}^S$ and the target graph $ Z_{E}^T$, $ Z_{O}^T$ as their heterophily and homophily representation. 
To this end, $\mathcal{L_A}$ utilizes the KL divergence loss between the source graph embeddings $ Z_{E}^S $, $ Z_{O}^S$ and the target graph $ Z_{E}^T $, $ Z_{O}^T $ and the target graph embeddings $ Z_{H}^T $, which can be formulated as:
\begin{equation}
\mathcal{L_A} = {KL}\big(Z_{E}^S \| Z_{E}^T\big) + {KL}\big(Z_{O}^S \| Z_{O}^T\big),
\end{equation}
In summary, the objective of RSGDA can be computed by:
\begin{equation}
    \mathcal{L} = \mathcal{L}_{CR} + \mu_1\mathcal{L}_{RE} + \mu_2\mathcal{L}_{A}.
\end{equation}
where $\mu_1$ and $\mu_2$ are trade-off hyper-parameters. The parameters of the framework are updated via backpropagation. A detailed description of our algorithm is provided in supplementary materials.

\subsection{Theoretical Analysis of RSGDA}

To further illustrate how RSGDA reduces the representation gap between domains, we examine the effect of the mixed graph filters on the latent feature distributions.  Specifically, we show that under appropriate spectral alignment, the outputs of any hypothesis $h \in \mathcal{H}$ exhibit negligible difference between the source and target domains when evaluated on the filtered and encoded representations. 

\begin{theorem}
    By using mixed graph filter for source and target graphs, for all input $z$, we have $|\mathbb{E}_{z\sim P_S^O}[h(z)]-\mathbb{E}_{z\sim P_T^O}[h(z)]| < \epsilon$, where $\epsilon$ denotes a small positive constant that can be made arbitrarily small as the optimization proceeds.
\end{theorem}

This result implies that our filter-based spectral normalization and structural alignment promote semantic consistency by minimizing distributional mismatch in the latent space via low-pass alignment, enabling any classifier $h$ to generalize across domains. As optimization proceeds, the bound $\epsilon$ tightens, validating the effectiveness of structural reconstruction and spectral filtering for unsupervised GDA.

\begin{definition}[Structural Difference]
We define the structural difference between graph $G_s$ and $G_t$ as 
$$
S(G_S,G_T)=\|L_O^S-L_O^T\|_2 + \|L_E^S-L_E^T\|_2.
$$ 
\end{definition}


\begin{theorem}[Domain Adaptation Bound for Graph Reconstruction]
Combining the empirical–source bound, the domain adaptation bound, and the divergence bound, we obtain that with probability at least $1-\delta$, for all $h\in\mathcal{H}$, we have
\begin{equation}
\begin{aligned}
\mathcal{R}_T(h)&
\;\le\;
\hat{\mathcal{R}}_S(h)
\;+\;
2\,\mathfrak{R}_n(\mathcal{H})
\;+\;
\sqrt{\frac{\ln(1/\delta)}{2n}}
\;\\
& +\;
C\,S(G_S,G_T)
\;+\;
\lambda,
\end{aligned}
\end{equation}
which is precisely inequality with a constant $C$. 
\end{theorem}

This bound characterizes the generalization ability of RSGDA in unsupervised GDA, where the first three terms reflect source estimation error, and the last term $S(G_s, G_t)$ measures spectral discrepancies between reconstructed homophilic and heterophilic structures. By minimizing $S(G_s, G_t)$ through homophily-agnostic reconstruction and spectral filtering, RSGDA effectively reduces the target risk upper bound.


\begin{table}[t]
\centering
\small
\scalebox{0.9}{\begin{tabular}{c|c|c|c|c}
\toprule[0.8pt]
Datasets        & \#Node  & \#Edge     & \#Homophily Ratio          & \#Label              \\ 
\midrule[0.8pt]
USA (U)      & 1,190 & 13,599 & 0.6978 & \multirow{3}{*}{4} \\
Brazil (B) & 131 & 1,038   &        0.4683             &                    \\
Europe  (E)    & 399 & 5,995      &   0.4048             &                    \\ 
\midrule[0.8pt]
ACMv9  (A)     & 9,360 & 15,556  &0.7798 & \multirow{3}{*}{5} \\
Citationv1 (C)  & 8,935 & 15,098    & 0.8598                    &                    \\
DBLPv7 (D)     & 5,484 & 8,117    &     0.8198                 &                    \\ 
\midrule[0.8pt]
ACM3 & 3,025 & 2,221,699 & 0.1034 &\multirow{2}{*}{3} \\
ACM4 & 4019 & 57,853   & 0.8391\\

\midrule[0.8pt]
Blog1 (B1)    & 2,300 & 33,471  & 0.3991 & \multirow{2}{*}{6} \\
Blog2 (B2)   & 2,896 & 53,836    & 0.4002                    &                    \\ 
\midrule[0.8pt]
Texas (TX)      & 183 & 325   &0.0614    & \multirow{3}{*}{5} \\
Cornell (CO)      &  183 &  298     &0.1220                  &                    \\
Wisconsin (WI)      & 251 & 515    &0.1703                     &                    \\ 
\midrule[0.8pt]
\end{tabular}}
\caption{Dataset Statistics.}
\label{tab:datasets}
\end{table}

\begin{table*}[!t]
\centering
\small
\scalebox{0.95}{\begin{tabular}{l|cccccc|cc|cc}
\toprule
Methods & U $\rightarrow$ B & U $\rightarrow$ E & B $\rightarrow$ U & B $\rightarrow$ E & E $\rightarrow$ U & E $\rightarrow$ B & A3 $\rightarrow$ A4 & A4 $\rightarrow$ A3 &  B1 $\rightarrow$ B2 & B2 $\rightarrow$ B1\\\midrule
GCN & 0.366 &	0.371 &	0.491 &	0.452 &	0.439 &	0.298 &	0.373 &	0.323 & 0.408 &	0.451 \\
DANN & 0.501 &	0.386 &	0.402 &	0.350 &	0.436 &	0.538  &0.362 &	0.325 & 0.409 &0.419 \\\midrule
DANE& 0.531 &	0.472 &	0.491 &	0.489 &	0.461 &	0.520 &	0.392 &	0.404 & 0.464 &	0.423\\ 
UDAGCN & 0.607 &	0.488 &	0.497 &	0.510 &	0.434 &	0.477 &	0.404 &	0.380 & 0.471 &	0.468  \\
ASN & 0.519 &	0.469 &	0.498 &	0.494 &	0.466 &	0.595 &0.418 &	0.409& 0.632 &0.524\\
EGI & 0.523 &0.451 &	0.417 &	0.454 &	0.452 &	0.588 &0.511 &	0.449 & 0.494 & 0.516 \\ 
GRADE-N & 0.550 &	0.457 &	0.497 &	0.506 &	0.463 &	0.588 &	0.449 &	0.461 & 0.567 &	0.541 \\
JHGDA  & {0.695} &0.519 &	0.511 &	{0.569} &	0.522 &	\underline{0.740} &	0.516 &	{0.537} & 0.619 &0.643\\
SpecReg & 0.481 &0.487 &0.513 &	{0.546} &	0.436 &	0.527 &0.526 &	0.518 &0.661 &0.631  \\
PA      & 0.621   &0.547    &0.543 &	0.516    &	0.506  &0.670 &0.619 &	0.610 & {{0.662}} &	{{0.654}}\\
{HGDA}   &\underline {0.721} &\underline {0.572} &\underline {0.569} &\underline {0.584} &\underline{0.570} & 0.721 &\underline{0.718} &\underline {0.698} &\underline {0.683} &\underline {0.677} \\\midrule
\textbf{RSGDA} & \bf{0.744} &\bf{0.591} &	\bf{0.571} &	\bf{0.609} &	\bf{0.581} &	\bf{0.740} &	\bf{0.726} &	\bf{0.732} & \bf{0.709} &\bf{0.693}\\

\bottomrule
\end{tabular}}
\caption{Cross-network node classification on the Airport, ACM and Blog network.}
\label{tab:airport_classification}
\end{table*}

\begin{table*}[!t]
\centering
\small
\scalebox{0.95}{\begin{tabular}{l|cccccc|cccc}
\toprule
Methods & A $\rightarrow$ D & D $\rightarrow$ A & A $\rightarrow$ C & C $\rightarrow$ A & C $\rightarrow$ D & D $\rightarrow$ C & CO $\rightarrow$ WI & TX $\rightarrow$ CO & TX $\rightarrow$ WI & WI $\rightarrow$ TX  \\\midrule
GCN & 0.632 &	0.578 &	0.675 &	0.635 &	0.666 &	0.654  &0.218 &0.384 &0.218 &	0.308 \\
DANN & 0.488 &0.436 &0.520 &	0.518 &	0.518 &	0.465  &0.242 &	0.318 &	0.277 &	0.341  \\\midrule
DANE& 0.664 &	0.619 &	0.642 &	0.653 &	0.661 &	0.709  &0.271 &	0.243 &	0.281 &	0.374\\ 
UDAGCN  & 0.684 &	0.623 &	0.728 &	0.663 &	0.712 &	0.645  &	0.624 &	0.260  &	0.249 &	{0.375}  \\
ASN & 0.729 &0.723 &0.752 &	0.678 &	0.752 &	0.754  &0.351 &	0.378 &	0.218 &	0.401 \\
EGI & 0.647 & 0.557 &0.676 &	0.598 &	0.662 &	0.652 &0.388 &	0.365  &	0.247 &	0.391 \\ 
GRADE-N & 0.701 &	0.660 &	0.736 &	0.687 &	0.722 &	0.687 &	0.348 &	0.364&	0.216 &	\underline{0.415} \\
JHGDA & 0.755 &0.737 &	0.814 &	{0.756} &	0.762 &	0.794 &	0.386 &	{0.407}&	\underline{0.281} &	{0.239} \\
SpecReg & 0.762 &0.654 &0.753 &	0.680 &	{0.768} &	0.727 &	{0.719} &	{0.214} &	{0.145} &	{0.255}\\
{PA} & {0.752} &	{{0.751}} &	{0.804} &	{0.768} &	{0.755} &	{0.780} & {0.279}     &	0.280 & {0.291}     &	{0.288}  \\ 
{HGDA} & \underline{0.791} &	\underline{0.756} &\underline{0.829} &\underline{0.787} &\underline{0.779} &\underline{0.799}  &	\underline{0.381}  &	\underline{0.303}   &	0.199  &	0.371 \\\midrule
\textbf{RSGDA} & \bf{0.812} &\bf{0.809} &	\bf{0.834} &	\bf{0.797} &	\bf{0.795} &\bf	{0.808} &	\bf{0.416} &	\bf{0.497} & \bf{0.425} &\bf{0.503}\\

\bottomrule
\end{tabular}}
\caption{Cross-network node classification on the Citation and WebKB network.}
\label{tab:citation_classification}
\end{table*}

\section{Experiments}

\subsection{Datasets}
To demonstrate the effectiveness of our approach on domain adaptation for node classification tasks, we evaluate it on four types of datasets: Airport~\cite{ribeiro2017struc2vec}, Citation~\cite{wu2020unsupervised}, Social~\cite{liu2024rethinking}, and WebKB~\cite{peigeom}.
The Airport dataset represents airport networks from three regions: the USA (U), Brazil (B), and Europe (E), where nodes denote airports and edges represent flight routes. The Citation dataset includes three citation networks: DBLPv8 (D), ACMv9 (A), and Citationv2 (C), with nodes as articles and edges as citation links. For the Social domain, we use two blog networks, Blog1 (B1) and Blog2 (B2), both extracted from the BlogCatalog dataset. The WebKB dataset contains widely used heterophilic graphs such as Texas (TX), Cornell (CO), and Wisconsin (WI). To further assess the impact of homophily distribution shifts, we curate a real-world dataset exhibiting significant homophily variation. We utilize two commonly referenced ACM datasets: ACM3 (A3), derived from the ACM Paper-Subject-Paper (PSP) network~\cite{fan2020one2multi}, and ACM4 (A4), extracted from the ACM2 Paper-Author-Paper (PAP) network~\cite{fu2020magnn}. These datasets exhibit inherently different distributions, making them ideal for evaluating domain adaptation performance.

\subsection{Baselines}
We compare RSGDA with representative GDA methods.
{GCN}\cite{kipf2016semi} introduces a first-order approximation of ChebNet.
{DANN}\cite{ganin2016domain} employs a 2-layer MLP and a gradient reversal layer for domain-invariant feature learning.
{DANE}\cite{zhang2019dane} aligns shared latent distributions across networks via adversarial regularization.
{UDAGCN}\cite{wu2020unsupervised} captures both local and global knowledge using dual GCNs under adversarial training.
{ASN}\cite{zhang2021adversarial} extracts domain-invariant features by disentangling domain-specific components.
{EGI}\cite{zhu2021transfer} maximizes ego-graph information to model structural transferability.
{GRADE-N}\cite{wu2023non} measures distribution shifts via graph subtree discrepancies.
{JHGDA}\cite{shi2023improving} leverages hierarchical pooling to exploit multi-level graph structures.
{SpecReg}\cite{you2022graph} integrates optimal transport and graph filters for regularized domain adaptation.
{PA}\cite{liu2024pairwise} mitigates structure and label shift by reweighting edges based on conditional dependence.
{HGDA}~\cite{fanghomophily} separates homophilic and heterophilic signals using distinct spectral filters.

\subsection{Performance Comparison}
The experimental results are reported in Table~\ref{tab:airport_classification} and Table~\ref{tab:citation_classification}, where the best and second-best results are marked in \textbf{bold} and \underline{underlined}, respectively.

Our proposed {RSGDA} consistently achieves state-of-the-art performance across all benchmarks, significantly surpassing prior methods. Specifically, {RSGDA} obtains the highest accuracy in every cross-domain scenario, with substantial improvements in both homophilic and heterophilic settings. For instance, on the challenging heterophilic WebKB datasets, RSGDA outperforms the second-best method by up to $8.6\%$ on WI $\rightarrow$ TX, demonstrating its ability to capture complex structure shifts.
Compared to the recent strong baseline HGDA, RSGDA shows an average gain of $2.30\%$ on Airport, $2.56\%$ on Citation, and $2.11\%$ on Blog datasets. Notably, it also improves upon \textbf{PA}~\citep{liu2024pairwise}, a leading 2024 method, by a large margin. These improvements validate the effectiveness of our spectral graph reconstruction and homophily-agnostic alignment strategy.
In conclusion, the consistent top performance across diverse benchmarks highlights RSGDA’s robustness in handling varying homophily levels and structural discrepancies between source and target graphs. 

\begin{figure}[t]
    \centering
    \includegraphics[width=.95\linewidth]{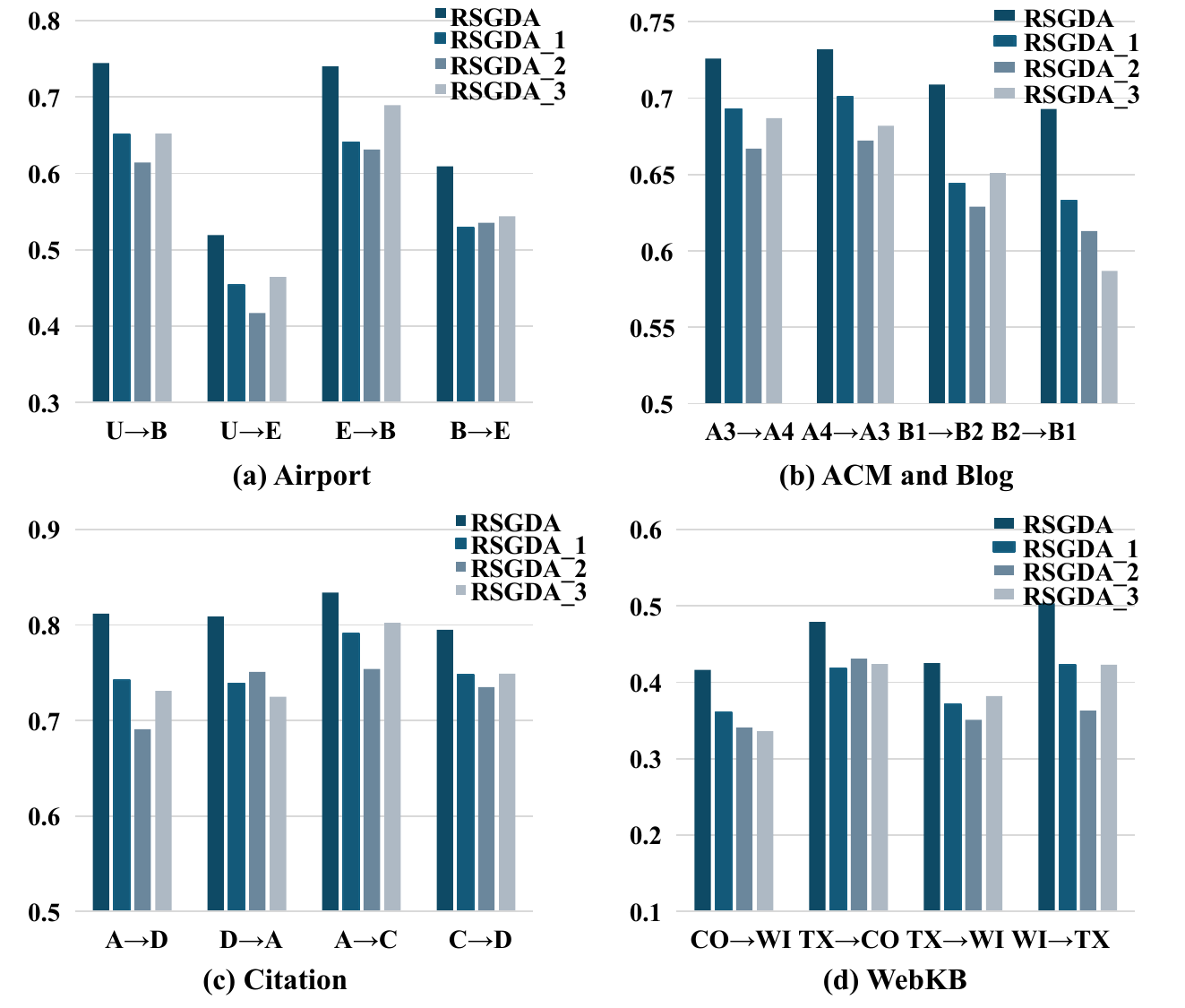}    
    \caption{The classification accuracy of RSGDA and its variants on four benchmark datasets.}
    \label{fig:ablation}
\end{figure}

\subsection{Ablation Study}
To validate the effectiveness of different components in our model, we compare RSGDA with its three ablated variants across five benchmarks, as shown in Figure~\ref{fig:ablation}.
{RSGDA$_1$}: RSGDA without the correlation reduction loss $\mathcal{L}_{CR}$, used to evaluate its role in mitigating representation collapse.
{RSGDA$_2$}: RSGDA without the feature reconstruction loss $\mathcal{L}_{RE}$, used to assess the impact of alleviate representation collapse.
{RSGDA$_3$}: RSGDA is applied to the original graphs, randomly separating them into two graphs without utilizing the proposed structure reconstruction module.

According to Figure~\ref{fig:ablation}, we can draw the following conclusions: (1) RSGDA consistently outperforms all its variants across all domains, demonstrating the overall effectiveness and rationality of the full model. (2) Removing either the correlation reduction loss or the reconstruction loss leads to performance degradation, highlighting the importance of regularization and feature recovery. (3) The most significant performance drop occurs in RSGDA$_3$, indicating that our structure reconstruction mechanism is critical in homophily-agnostic representations and mitigating domain shift.

\begin{figure}[t]
    \centering
    \subfigure[U $\rightarrow$ B]{
        \includegraphics[width=.47\linewidth]{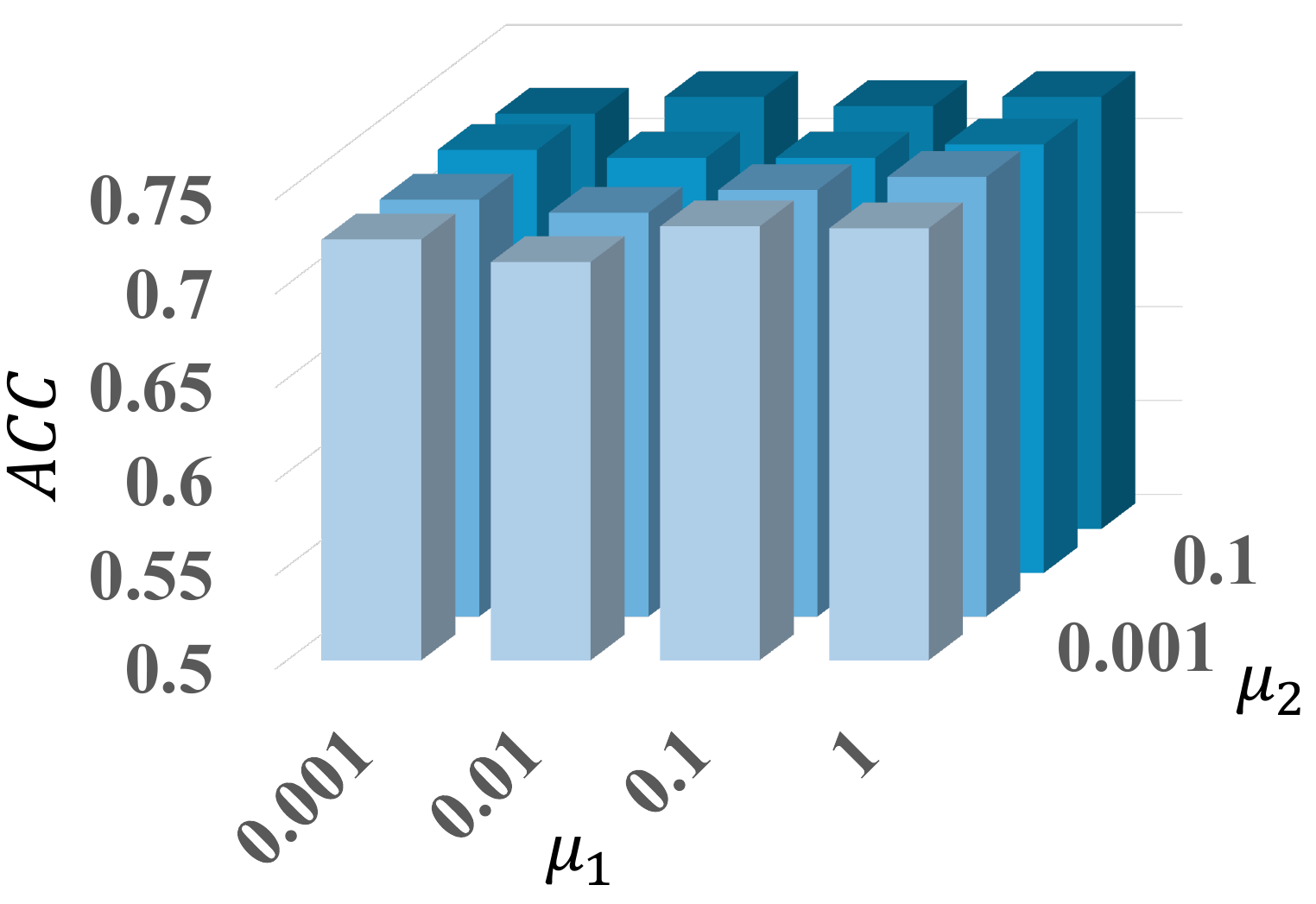}
    }
    \subfigure[TX $\rightarrow$ WI]{
        \includegraphics[width=.47\linewidth]{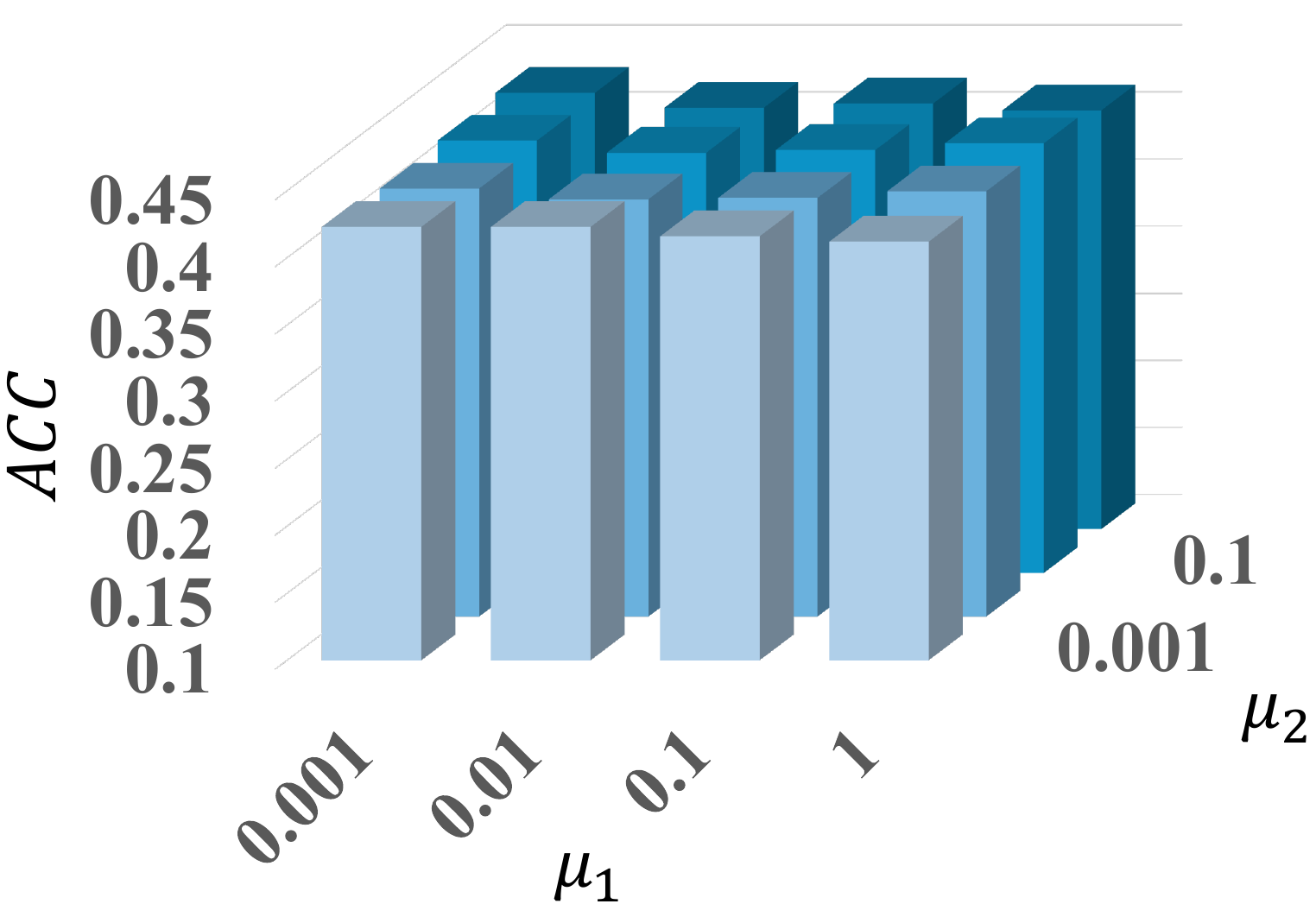}
    }
    \caption{The influence of parameters $\mu_1$ and $\mu_2$ on Airport and WebKB datasets.}
    \label{Parameter Analysis}
\end{figure}

\subsection{Parameter Analysis}
In this section, we analyze the sensitivity of hyper-parameters $\mu_1$ and $\mu_2$ of our method on Airport and WebKB datasets. 
First, we test the performance with different  $\mu_1$ and $\mu_2$. As shown in Figure \ref{Parameter Analysis}, RSGDA has competitive performance on a large range of values, which suggests the stability of our method. For a more detailed analysis and result, refer to the Appendix.

\section{Conclusion}
We propose RSGDA, a homophily-agnostic framework for graph domain adaptation. By reconstructing both homophilic and heterophilic graphs and applying an adaptive filter, RSGDA extracts complementary structural signals and mitigates distribution shift without label supervision. Extensive experiments and theoretical analysis validate its robustness and effectiveness across diverse homophily settings.

\section{Acknowledgements}
This work is supported by the Natural Sciences and Engineering Research Council of Canada (NSERC), Discovery Grants program.

\bibliography{aaai2026}
\maketitle
\newpage
\appendix

\subsection{Experimental Setup}
The experiments are implemented in the PyTorch platform using an Intel(R) Xeon(R) Silver 4210R CPU @ 2.40GHz, and GeForce RTX A5000 24G GPU.
Technically, two layers GCN is built and we train our model by utilizing the Adam~\citep{adam} optimizer with learning rate ranging from 0.0001 to 0.0005. In order to prevent over-fitting, we set the dropout rate to 0.5.
In addition, we set weight decay $\in \left\{1e-4, \cdots, 5e-3 \right\}$ and $l$ $\in \left\{1, \cdots, 10 \right\}$ for $k$NN graph. For fairness, we use the same parameter settings for all the cross-domain node classification methods in our experiment, except for some special cases. For GCN, UDA-GCN, and JHGDA the GCNs of both the source and target networks contain two hidden layers $(L = 2)$ with structure as $128-16$. The dropout rate for each GCN layer is set to $0.3$.
We repeatedly train and test our model for five times with the same partition of dataset and then report the average of ACC.

\section{Model Efficient Experiment}
\textbf{Model Efficient Experiment}: 
To further evaluate the efficiency of RSGDA, Table \ref{tab:time datasets} presents a running time comparison across various algorithms. We also compared the GPU memory usage of common baselines, including UDAGCN and the recent state-of-the-art methods JHGDA, PA, and SpecReg, which align graph domain discrepancies through different ways. As shown in Table, the evaluation results on airport and citation dataset further demonstrate that our method achieves superior performance with tolerable computational and storage overhead~\cite{fang2022structure,fang2026SAGA}.

\begin{table*}[!t]
\centering
\small
\scalebox{0.8}{\begin{tabular}{|c|c|c|c|c|c|}
\toprule[0.8pt]
Dataset   &Method  &Training Time (Normalized w.r.t. UDAGCN)   &Memory Usage (Normalized w.r.t. UDAGCN)    & Accuracy($\%$)         \\ 
\midrule[0.8pt]      
\multirow{6}{*}{U$\rightarrow$B}  & UDAGCN      & 1 & 1   & 0.607 \\
                          & JHGDA     & 1.314 & 1.414    & \underline{0.695} \\
                          & PA      & 0.498 & {0.517}   & 0.481 \\
                          & SpecReg & \underline{0.463}   & \underline{0.493} & 0.621 \\
                          &$RSGDA$ & \bf{0.412} & \bf0.209     & \bf{0.511} \\
\midrule[0.8pt]

\multirow{6}{*}{U$\rightarrow$E}  & UDAGCN      & 1 & 1   & 0.472 \\
                          & JHGDA     & 1.423 & 1.513    & 0.519 \\
                          & PA      & \underline{0.511} & 0.509   & \underline{0.547} \\
                          &SpecReg & \bf0.517 & \underline{0.503}   & 0.487 \\
                          & $RSGDA$ & \bf0.217 & \b{0.213}   & \bf0.472 \\
\midrule[0.8pt]
\multirow{6}{*}{B$\rightarrow$E}  & UDAGCN      & 1 & 1   & 0.497 \\
                          & JHGDA     & 1.311 & 1.501    & \underline{0.569} \\
                          & PA      & 0.502 & \underline{0.497}   & 0.516 \\
                          &SpecReg & \underline{0.407} & 0.503   & {0.546} \\
                          & $RSGDA$ & \bf{0.209} & \bf{0.271}   & \bf{0.484} \\
\midrule[0.8pt]
\multirow{6}{*}{A$\rightarrow$D}  & UDAGCN      & 1 & 1   & 0.510 \\
                          & JHGDA     & 1.311 & 1.501    & 0.569 \\
                          & PA      & 0.502 & \underline{0.497}   & 0.562 \\
                          &SpecReg & \underline{0.407} & 0.503   & \underline{0.536} \\
                          & $RSGDA$ & \bf{0.309} & \bf{0.313}   & \bf{0.791} \\
\midrule[0.8pt]
\multirow{6}{*}{A$\rightarrow$C}  & UDAGCN      & 1 & 1   & 0.510 \\
                          & JHGDA     & 1.311 & 1.501    & 0.569 \\
                          & PA      & 0.502 & \underline{0.497}   & 0.562 \\
                          &SpecReg & \underline{0.407} & 0.503   & \underline{0.536} \\
                          & $RSGDA$ & \bf{0.309} & \bf{0.313}   & \bf{0.326} \\
\midrule[0.8pt]
\multirow{6}{*}{C$\rightarrow$D}  & UDAGCN      & 1 & 1   & 0.510 \\
                          & JHGDA     & 1.311 & 1.501    & 0.570 \\
                          & PA      & 0.502 & \underline{0.497}   & 0.562 \\
                          &SpecReg & \underline{0.407} & 0.503   & \underline{0.536} \\
                          & $RSGDA$ & \bf{0.309} & \bf{0.313}   & \bf{0.326} \\
\midrule[0.8pt]

\end{tabular}}
\caption{Comparison of Training Time, Memory Usage, and Accuracy on Airport datset.}
\label{tab:time datasets}
\end{table*}

\label{Model efficient experiment}

\subsection{Additional Parameter Analysis}

\begin{figure}[t]
    \centering
    \subfigure[Airport]{
        \includegraphics[width=.47\linewidth]{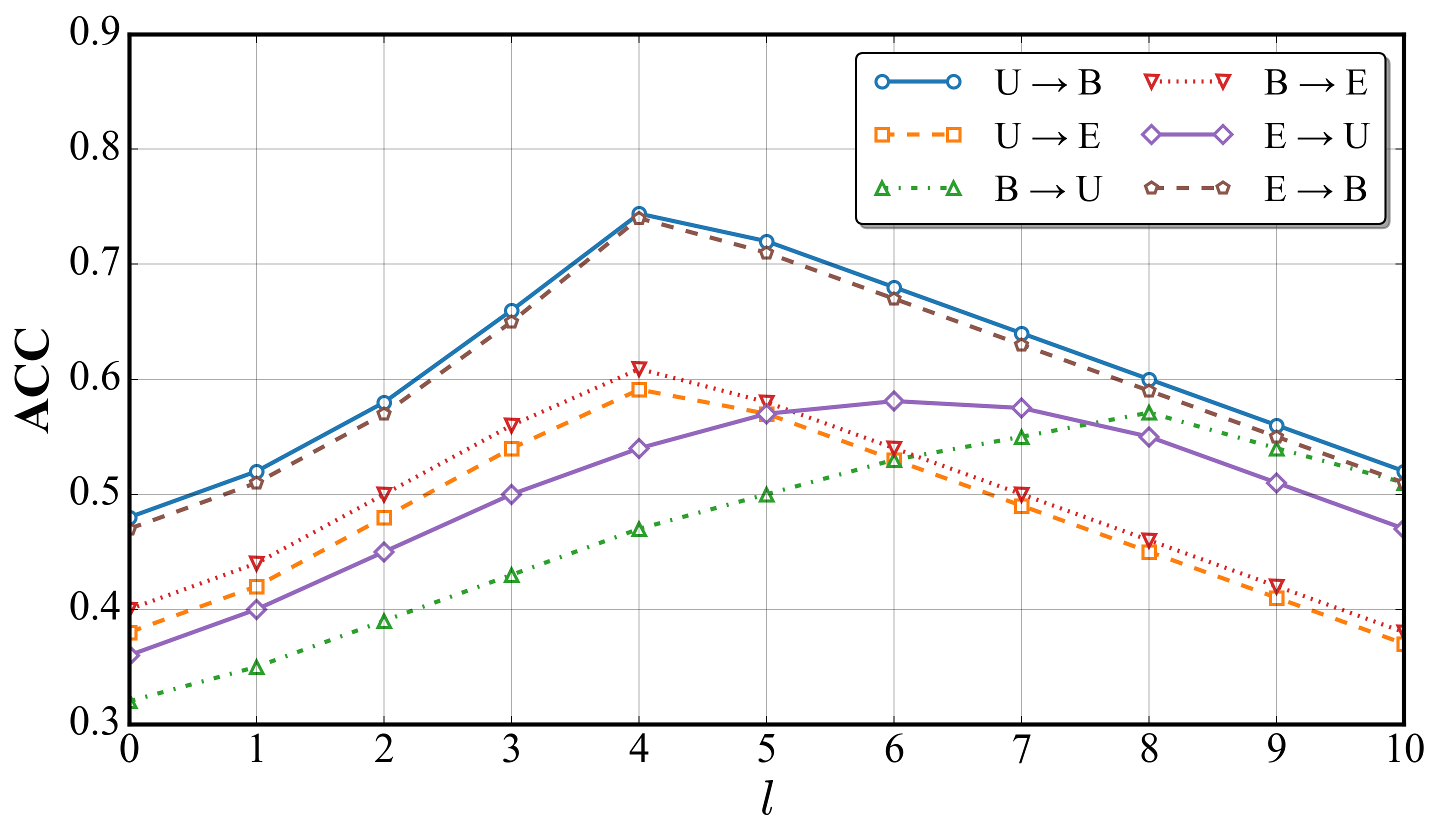}
    }
    \subfigure[WebKB]{
        \includegraphics[width=.47\linewidth]{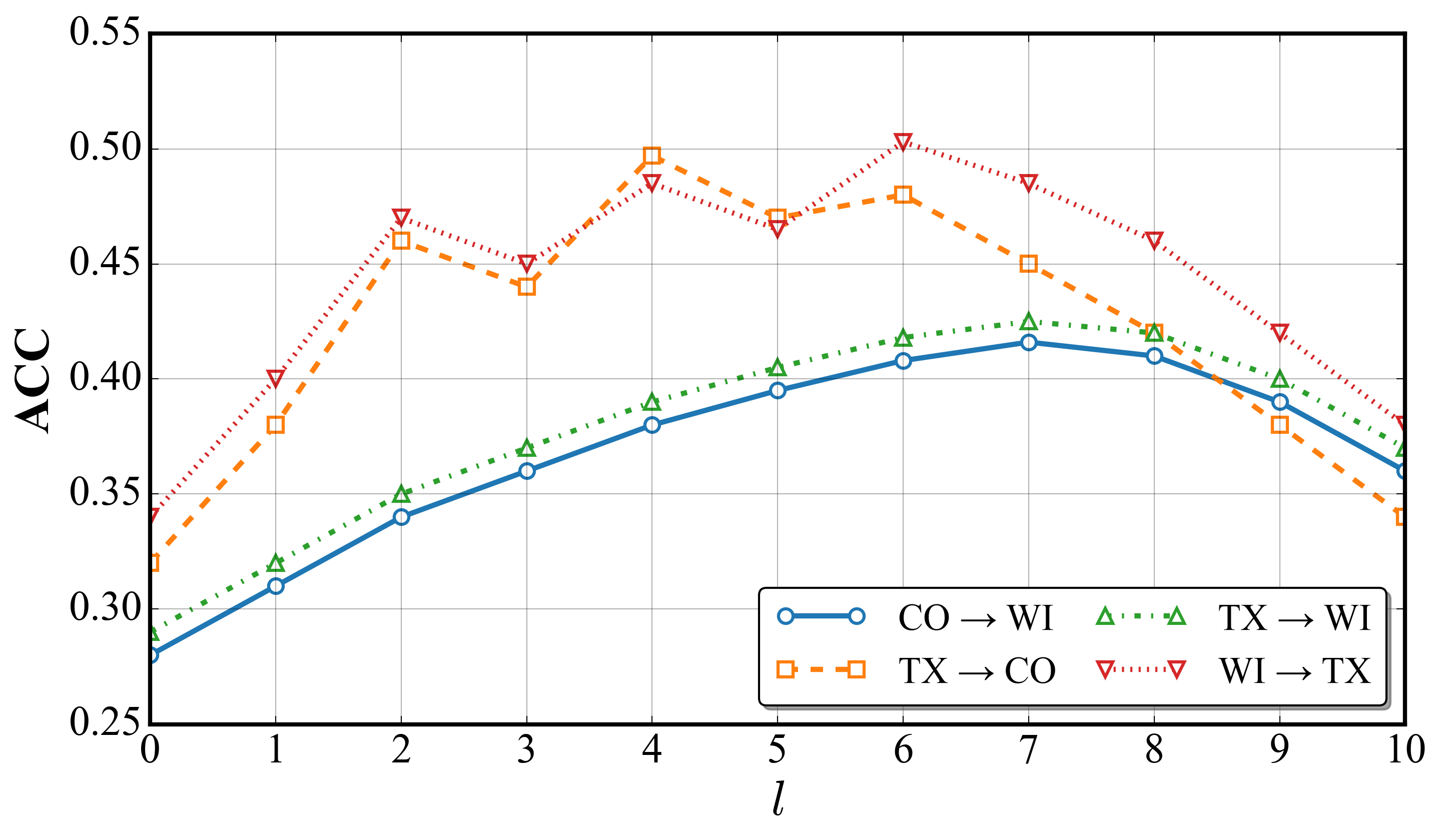}
    }
    \caption{The influence of parameters $l$ on Citation and Airport dataset.}
    \label{Parameter Analysis}
\end{figure}

We conduct a sensitivity analysis on the hyperparameter $l$, which denotes the number of steps in our structure reconstruction or filtering process. As shown in Figure~\ref{Parameter Analysis}, we evaluate the model’s performance across multiple cross-domain settings on both homophilic (top plot) and heterophilic (bottom plot) graph datasets.

Across all settings, we observe that increasing $l$ initially improves the accuracy (ACC), as it allows the model to better capture structural information and smooth signals over the graph. Performance generally peaks around $l=4$ or $l=5$, after which further increasing $l$ leads to performance degradation. This trend aligns with our observation in Fig. 3 (a) in the main paper.

Specifically, for homophilic graphs (e.g., U→B, E→B), the model achieves a clear performance peak at $l=4$, while for heterophilic datasets (e.g., WI→TX, TX→CO), the optimal $l$ varies slightly but still follows a similar concave pattern. These results confirm that moderate levels of propagation or reconstruction (i.e., intermediate $l$) are critical for balancing structural refinement and information preservation in unsupervised graph domain adaptation. This trend also aligns with our observation in Fig. 3 (b) in the main paper.

\subsection{Hyperparameter Setting}

To ensure a fair and robust evaluation of our method across diverse transfer scenarios, we conduct fine-grained hyperparameter selection for each domain adaptation task. Table~\ref{tab:hyperparams_flat} summarizes the selected values of the regularization coefficients $\mu_1$, $\mu_2$ and the propagation depth $l$ used during training.

We observe that the optimal choice of hyperparameters is closely linked to the structural and statistical characteristics of each dataset. For instance, in the Airport domain, which models flight connections between regions with moderately varying homophily (e.g., 0.69 for USA vs. 0.40 for Europe), relatively higher $\mu$ values and deeper propagation (e.g., $l=4$ for U$\rightarrow$B, E$\rightarrow$U) are preferred. This suggests that cross-region graph structures benefit from stronger regularization and richer multi-hop aggregation to bridge distributional discrepancies.

In contrast, tasks in the Citation domain (e.g., A$\rightarrow$C, D$\rightarrow$A) show consistent performance with relatively uniform and low $\mu$ settings (0.1–0.5) and moderate $l$ values (3–4). This is likely due to the shared scholarly domain and relatively consistent citation patterns across DBLPv8, ACMv9, and Citationv2, which reduce the need for aggressive adaptation or deep structural propagation.

For the Social domain (Blog1 and Blog2), where the homophily ratios are notably low (~0.4), shallow propagation depths ($l=2$ or $3$) suffice. This is consistent with prior observations in heterophilic graphs, where deep message passing tends to introduce noise. Regularization terms are also kept small, indicating that the alignment loss needs to be carefully balanced to avoid overfitting on noisy interactions.

In the WebKB benchmark (TX, CO, WI), known for its heterophilic structure (e.g., homophily as low as 0.06 in Texas), we adopt deeper propagation for certain transfer tasks such as WI$ \ rightarrow$TX ($l=8$) or CO $\rightarrow$ WI ($l=7$). These choices reflect the need for extended receptive fields to capture relevant but indirect relationships. Meanwhile, regularization remains mild to prevent over-constraining model flexibility.

Finally, the curated real-world pair A3$\rightarrow$A4 and A4$\rightarrow$A3, which simulate severe distribution shifts across semantic relations (PSP vs. PAP graph structures), are assigned balanced regularization weights ($\mu_1 = \mu_2 = 0.1$) and moderate $l$ values (4–5). This empirically confirms that homophily variation can be effectively managed through a moderate level of propagation coupled with light regularization.

\begin{table}[!t]
\centering
\small
\caption{Hyperparameter settings for each transfer task.}
\begin{tabular}{lccc}
\toprule[0.8pt]
\textbf{Transfer Task} & $\mu_1$ & $\mu_2$ & $l$ \\
\midrule[0.8pt]
U $\rightarrow$ B   & 0.5 & 0.5 & 4 \\
U $\rightarrow$ E   & 0.1 & 0.1 & 2 \\
B $\rightarrow$ U   & 0.1 & 0.1 & 4 \\
B $\rightarrow$ E   & 0.5 & 0.1 & 3 \\
E $\rightarrow$ U   & 0.5 & 0.5 & 4 \\
E $\rightarrow$ B   & 0.5 & 0.5 & 4 \\
\midrule[0.8pt]
A3 $\rightarrow$ A4 & 0.1 & 0.1 & 5 \\ 
A4 $\rightarrow$ A3 & 0.1 & 0.1 & 4 \\ 
\midrule[0.8pt]
B1 $\rightarrow$ B2 & 0.1 & 0.1 & 2 \\
B2 $\rightarrow$ B1 & 0.1 & 0.1 & 3 \\
\midrule[0.8pt]
A $\rightarrow$ D   & 0.1 & 0.1 & 3 \\
D $\rightarrow$ A   & 0.1 & 0.1 & 4 \\
A $\rightarrow$ C   & 0.5 & 0.5 & 4 \\
C $\rightarrow$ A   & 0.1 & 0.1 & 3 \\
C $\rightarrow$ D   & 0.1 & 0.1 & 4 \\
D $\rightarrow$ C   & 0.1 & 0.1 & 4 \\
\midrule[0.8pt]
CO $\rightarrow$ WI & 0.1 & 0.1 & 7 \\ 
TX $\rightarrow$ CO & 0.1 & 0.1 & 4 \\ 
TX $\rightarrow$ WI & 0.1 & 0.1 & 6 \\
WI $\rightarrow$ TX & 0.1 & 0.1 & 8 \\
\bottomrule
\end{tabular}
\label{tab:hyperparams_flat}
\end{table}

\section{Theoretical Analysis}

\begin{theorem}
    By using mixed graph filter for source and target graphs, for all input $z$, we have $|\mathbb{E}_{z\sim P_S^O}[h(z)]-\mathbb{E}_{z\sim P_T^O}[h(z)]| < \epsilon$, where $\epsilon$ denotes a small positive constant that can be made arbitrarily small as the optimization proceeds.
\label{one}
\end{theorem}

The mixed graph filter is applied identically to $G_S$ and $G_T$, so that both domains are mapped into latent spaces with comparable spectral support.  Let $\mathcal{E}_\theta$ denote the shared encoder.  It is standard in theoretical analyses to assume that $\mathcal{E}_\theta$ is $\gamma$-Lipschitz.  In practice, $\gamma$ can be enforced to be $1$ by applying weight or spectral normalization to the final multi-layer perceptron.  Consequently, for any inputs $x,x'$ we have 

\begin{equation}
    ||\mathbb{E}_\theta (x)-\mathbb{E}_\theta (x')||_2 \le ||x-x'||_2
\end{equation}

Denote by $Z_S^O\sim P_S^O$, $Z_T^O\sim P_T^O$, $Z_S^E\sim P_S^E$ and $Z_T^E\sim P_T^E$ the random variables representing the latent features obtained from the source and target graphs via the homophilic and heterophilic paths, respectively.  The alignment loss in RSGDA is defined as $\mathcal{L_A}=KL(Z_E^S||Z_E^T)+KL(Z_O^S||Z_O^T)$.  Interpreting $Z_S^E$ and $Z_T^E$ as samples from the distributions $P_S^E$ and $P_T^E$, the loss becomes 

\begin{equation}
\mathcal{L_A}=KL(P_E^S||P_E^T)+KL(P_O^S||P_O^T)
\end{equation}

Minimizing $\mathcal{L}_A$ forces these distributions to converge; in particular, if $\mathcal{L}_A\to 0$ we have $D_{\mathrm{KL}}(P_S^O\|P_T^O)\to 0$ and $D_{\mathrm{KL}}(P_S^E\|P_T^E)\to 0$. By Pinsker’s inequality, 

\begin{equation}
    ||P-Q||_{TV} \le \sqrt{2KL(P||Q)}
\end{equation}

where $||\cdot||_{TV}$ denotes the total variation distance, defined by $||P-Q||_{TV}=sup_{A}|P(A)-Q(A)|$. When $\mathcal{L_A} \rightarrow 0$, we have 

\begin{equation}
\mathcal{L_A}=KL(P_E^S||P_E^T)+KL(P_O^S||P_O^T)\rightarrow 0
\end{equation}

Since the Kullback–Leibler divergence is non‑negative, it follows that $KL(P_E^S||P_E^T) \rightarrow 0, KL(P_O^S||P_O^T)\rightarrow 0$. Hence, 

\begin{equation}
    ||P_O^S-P_O^T|| \rightarrow 0,||P_E^S-P_E^T|| \rightarrow 0.
\end{equation}

This implies that the source and target latent representations in both the homophilic and heterophilic subspaces become indistinguishable in total variation distance.  Moreover, because the encoder is Lipschitz continuous, we have 

\begin{equation}
    |\mathbb{E}_{z\sim P_S^O}[h(z)]-\mathbb{E}_{z\sim P_T^O}[h(z)]| \le ||P_O^S-P_O^T|| < \epsilon
\end{equation}

Therefore, as $\mathcal{L}_A$ approaches zero, the difference in expected classifier outputs on source and target representations can be made arbitrarily small; the same reasoning holds for the heterophilic path.  This completes the proof.

\begin{theorem}[Domain Adaptation Bound for Graph Reconstruction]
Combining the empirical–source bound, the domain adaptation bound, and the divergence bound, we obtain that with probability at least $1-\delta$, for all $h\in\mathcal{H}$, we have
\begin{equation}
\begin{aligned}
\mathcal{R}_T(h)&
\;\le\;
\hat{\mathcal{R}}_S(h)
\;+\;
2\,\mathfrak{R}_n(\mathcal{H})
\;+\;
\sqrt{\frac{\ln(1/\delta)}{2n}}
\;\\
& +\;
C\,S(G_S,G_T)
\;+\;
\lambda,
\end{aligned}
\end{equation}
which is precisely inequality with a constant $C$. 
\end{theorem}

Let $G_S$ be a labeled source graph with $n$ observed nodes $\{(x_i,y_i)\}_{i=1}^n$ sampled i.i.d. from a distribution $P_S$, and let $G_T$ be an unlabeled target graph with distribution $P_T$.  Let $\mathcal{H}$ denote a hypothesis class of measurable 0-1 classifiers.  For any $h\in\mathcal{H}$, define the empirical risk on the source graph as 

\begin{equation}
    \hat{\mathcal{R}}_S(h)
=\frac{1}{n}\sum_{i=1}^n \mathbb{I}\{h(x_i)\neq y_i\}
\end{equation}

the true risks as

\begin{equation}
    \mathcal{R}_S(h)=\Pr_{(x,y)\sim P_S}[h(x)\neq y],
\mathcal{R}_T(h)=\Pr_{(x,y)\sim P_T}[h(x)\neq y],
\end{equation}

and let $\lambda := \inf_{h\in\mathcal{H}} (\mathcal{R}_S(h)+\mathcal{R}_T(h))$ be the joint optimal error. A standard result \cite{bartlett2002rademacher} states that for any $\delta\in(0,1)$, with probability at least $1-\delta$ over the sampling of the source data, the following holds uniformly for all $h\in\mathcal{H}$:

\begin{equation}
    \mathcal{R}_S(h)
\;\le\;
\hat{\mathcal{R}}_S(h)
\;+\;
2\,\mathfrak{R}_n(\mathcal{H})
\;+\;
\sqrt{\frac{\ln(1/\delta)}{2n}},
\end{equation}

where $\mathfrak{R}_n(\mathcal{H})$ is the Rademacher complexity of $\mathcal{H}$ with respect to the source distribution. \citet{ben2006analysis} established that for any pair of distributions $P_S,P_T$ on $\mathcal{X}$ and any $h\in\mathcal{H}$,

\begin{equation}
    \mathcal{R}_T(h)
\;\le\;
\mathcal{R}_S(h)
\;+\;
d_{\mathcal{H}}(P_S,P_T)
\;+\;
\lambda,
\end{equation}

where $d_{\mathcal{H}}$ is the $\mathcal{H}$‑divergence:

\begin{equation}
\begin{aligned}
        d_{\mathcal{H}}(P_S,P_T)
=
2\,\sup_{h,h'\in\mathcal{H}}
\Bigl|\Pr_{x\sim P_S}[h(x)\neq h'(x)] - \Pr_{x\sim P_T}[h(x)\neq h'(x)]\Bigr|
\end{aligned}
\end{equation}

RSGDA does not operate directly on raw node features $x$, but on latent representations obtained through graph reconstruction, mixed filtering, and encoding.  Let $Z_S$ and $Z_T$ denote the random latent vectors produced from the source and target graphs respectively. Concretely, recall that RSGDA reconstructs homophilic and heterophilic graphs $(L_O^S,L_O^T,L_E^S,L_E^T)$, applies a low‑pass filter $F_O(L)=(I-\frac{1}{2}L)^k$ and a high‑pass filter $F_E(L)=(\frac{1}{2}L)^k$ to the node features $X$, and then uses a $\gamma$-Lipschitz encoder $\mathcal{E}_\theta$.  We thus obtain latent distributions 

\begin{equation}
    P_O^S = \text{Law}\bigl(\mathcal{E}_\theta(F_O(L_O^S)X)\bigr),
P_O^T = \text{Law}\bigl(\mathcal{E}_\theta(F_O(L_O^T)X)\bigr),
\end{equation}

and analogous distributions $P_E^S,P_E^T$ for the heterophilic path.  The overall latent distribution is the mixture of the homophilic and heterophilic components. A basic result shows that the domain divergence $d_{\mathcal{H}}(P_S,P_T)$ can be bounded in terms of the total variation distance of these latent distributions:

\begin{equation}
\label{premier}
    d_{\mathcal{H}}(P_S,P_T)
\;\le\;
2\,\bigl\|P_S^O - P_T^O\bigr\|_{TV}
\;+\;
2\,\bigl\|P_S^E - P_T^E\bigr\|_{TV}.
\end{equation}

The factor 2 comes from the fact that $\mathcal{H}$ divergence is bounded by twice the TV distance and that we consider both homophilic and heterophilic parts. We now need to relate the TV distances to graph reconstruction.  Two tools are employed:

1. Pinsker’s inequality: For any distributions $P,Q$, $\|P-Q\|_{TV}\le \sqrt{2 D_{\mathrm{KL}}(P\|Q)}$.  In RSGDA, small $\mathcal{L}_A$ implies both $\|P_S^O - P_T^O\|_{TV}$ and $\|P_S^E - P_T^E\|_{TV}$ are small. We have already prove that in Theorem \ref{one}.

2. Spectral difference to latent difference:  Consider the low‑pass path.  We have

\begin{equation}
    Z_O^S = \mathcal{E}_\theta\bigl( (I-\tfrac{1}{2}L_O^S)^k X \bigr), 
Z_O^T = \mathcal{E}_\theta\bigl( (I-\tfrac{1}{2}L_O^T)^k X \bigr).
\end{equation}

Because $\mathcal{E}_\theta$ is 1‑Lipschitz (after normalization) and $\|X\|_2\le C_X$, the Frechet derivative of the matrix function $L\rightarrow (I-\frac{1}{2}L)^k$ implies that

\begin{equation}
\label{twice}
    \bigl\|Z_O^S - Z_O^T\bigr\|_F
\;\le\;
\frac{k}{2}\,C_X\,
\bigl\|L_O^S - L_O^T\bigr\|_2
\;=\;
\alpha\,\bigl\|L_O^S - L_O^T\bigr\|_2,
\end{equation}

Similarly, for the heterophilic path,

\begin{equation}
    \bigl\|Z_E^S - Z_E^T\bigr\|_F
\;\le\;
\frac{k}{2^k}\,C_X\,
\bigl\|L_E^S - L_E^T\bigr\|_2
\;=\;
\beta\,\bigl\|L_E^S - L_E^T\bigr\|_2.
\end{equation}

The key insight is that the difference between latent representations is controlled by the spectral differences of the reconstructed Laplacians.  RSGDA’s graph reconstruction is designed to minimize these differences, thereby shrinking the equations above. Under a mild smoothness assumption on the encoder outputs, it follows from standard results in distributional stability that the total variation distance between the distributions of $P_O^S$ and $P_O^T$ can be bounded by a constant multiple of $\|Z_O^S - Z_O^T\|_F$.  Thus we obtain

\begin{equation}
\begin{aligned}
    \label{dernier}
    \|P_O^S - P_O^T\|_{TV}
\;\le\;
C'_1\,\bigl\|L_O^S - L_O^T\bigr\|_2,\\
\|P_E^S - P_E^T\|_{TV}
\;\le\;
C'_2\,\bigl\|L_E^S - L_E^T\bigr\|_2,
\end{aligned}
\end{equation}

for some universal constants $C'_1,C'_2$ depending on $k,C_X$. Combining equation \ref{premier} with \ref{dernier} and the definition of the structural difference $S(G_S,G_T)=\|L_O^S-L_O^T\|_2 + \|L_E^S-L_E^T\|_2$ yields

\begin{equation}
    d_{\mathcal{H}\Delta\mathcal{H}}(P_S,P_T)
\;\le\;
C\,S(G_S,G_T),
\end{equation}

for $C=2(C'_1+C'_2)$. Combining the empirical–source bound, the domain adaptation bound, and the divergence bound, we obtain that with probability at least $1-\delta$, for all $h\in\mathcal{H}$, we have

\begin{equation}
\begin{aligned}
    \mathcal{R}_T(h)
\;&\le\;
\hat{\mathcal{R}}_S(h)
\;+\;
2\,\mathfrak{R}_n(\mathcal{H})
\;+\;
\sqrt{\frac{\ln(1/\delta)}{2n}} 
\;\\
\;&+\;  
C\,S(G_S,G_T)
\;+\;
\lambda,
\end{aligned}
\end{equation}

\textbf{Proof of equation \ref{twice}}. Let us define the filter function as $f(L) = \left(I - \frac{1}{2}L\right)^k$. The Frechet derivative is given by:

\begin{equation}
    Df_L(\Delta L) = -\frac{k}{2} \left(I - \frac{1}{2}L\right)^{k-1} \cdot \Delta L.
\end{equation}

Thus, its operator norm satisfies $\|Df_L\| \le \frac{k}{2}.$ Applying this bound, we obtain: 

\begin{equation}
\begin{aligned}
        \|Z_O^S - Z_O^T\|_F &= \|f(L_O^S)X - f(L_O^T)X\|_F \\
        &\le \frac{k}{2} \cdot \|X\|_2 \cdot \|L_O^S - L_O^T\|_2.
\end{aligned}
\end{equation}

Letting $C_X = \|X\|_2$, we conclude:

\begin{equation}
    \|Z_O^S - Z_O^T\|_F \le \alpha \cdot \|L_O^S - L_O^T\|_2, \quad 
\text{where } \alpha = \frac{k}{2} C_X.
\end{equation}

\end{document}